\documentclass[11pt]{llncs}

\usepackage{fullpage}
\usepackage{amssymb}
\usepackage{amsmath}
\usepackage[dvipsnames]{xcolor}
\usepackage{breakcites}
\usepackage{mdframed}
\usepackage{algorithm}
\usepackage{algpseudocode}
\usepackage[colorlinks=true,citecolor=NavyBlue,linkcolor=Black]{hyperref}
\usepackage{nicefrac}
\usepackage{csquotes}
\usepackage{mathtools}
\usepackage{xspace}
\usepackage{mleftright}
\usepackage[most]{tcolorbox}
\usepackage[adversary,sets,asymptotics,landau,lambda,operators,logic]{cryptocode}
\allowdisplaybreaks
\makeatletter
\spnewtheorem{claim}[theorem]{Claim}{\bfseries}{\itshape}
\let\c@lemma\relax
\spnewtheorem{lemma}[theorem]{Lemma}{\bfseries}{\itshape}
\let\c@corollary\relax
\spnewtheorem{corollary}[theorem]{Corollary}{\bfseries}{\itshape}
\let\c@proposition\relax
\spnewtheorem{proposition}[theorem]{Proposition}{\bfseries}{\itshape}
\spnewtheorem{assumption}[definition]{Assumption}{\bfseries}{\itshape}
\makeatother
\usepackage{tikz}
\usetikzlibrary{decorations.pathreplacing}
\usetikzlibrary{patterns}

\newcommand{\tildebigO}[1]{\ensuremath{\mathcal{\tilde O}\left(#1\right)}\xspace}

\newcommand{\lattice}{\ensuremath{\mathcal{L}}\xspace}
\newcommand{\latticedim}{\ensuremath{n}\xspace}
\newcommand{\solbound}{\ensuremath{\beta}\xspace}
\newcommand{\samplesize}{\ensuremath{m}\xspace}
\newcommand{\domainsize}{\ensuremath{q}\xspace}
\newcommand{\approxfactor}{\ensuremath{\gamma}\xspace}
\newcommand{\mindist}{\ensuremath{\lambda_\latticedim(\lattice)}\xspace}

\newcommand{\ssdiff}{\ensuremath{\bigtriangleup}\xspace}

\newcommand{\wei}{\ensuremath{w}\xspace}
\newcommand{\dist}{\ensuremath{d}\xspace}
\newcommand{\ham}{\ensuremath{\mathsf{HAM}}\xspace}
\newcommand{\inlen}{\ensuremath{\ell}}

\newcommand{\rfam}{\ensuremath{\mathcal{R}}\xspace}
\newcommand{\statpar}{\ensuremath{k}\xspace}

\newcommand{\hfam}{\ensuremath{\mathcal{H}}\xspace}
\renewcommand{\sample}{\ensuremath{\mathsf{Sample}}\xspace}
\newcommand{\hash}{\ensuremath{\mathsf{Hash}}\xspace}
\newcommand{\eval}{\ensuremath{\mathsf{Eval}}\xspace}

\newcommand{\efam}{\ensuremath{\mathcal{E}}\xspace}
\newcommand{\efunc}{\ensuremath{f}\xspace}
\newcommand{\encode}{\ensuremath{\mathsf{Encode}}\xspace}
\newcommand{\decode}{\ensuremath{\mathsf{Decode}}\xspace}

\newcommand{\enclen}{\ensuremath{\mathsf{Len}_\efam}\xspace}
\newcommand{\decdiff}{\ensuremath{\mathsf{Diff}}\xspace}
\newcommand{\removeZ}{\ensuremath{\mathsf{Peel}}\xspace}

\newcommandx*{\pcas}[2][1=\ ,2=\ ]{#1\highlightkeyword[#2]{as}}
\newcommandx*{\pcor}[2][1=\ ,2=\ ]{#1\highlightkeyword[#2]{or}}
\newcommandx*{\pcand}[2][1=\ ,2=\ ]{#1\highlightkeyword[#2]{and}}

\title{Property-Preserving Hash Functions from Standard Assumptions}
\author{
	Nils Fleischhacker\inst{1}\thanks{\texttt{mail@nilsfleischhacker.de}. Funded by the Deutsche Forschungsgemeinschaft (DFG, German Research Foundation) under Germany's Excellence Strategy - EXC 2092 CASA - 390781972.}
	\and
	Kasper Green
        Larsen\inst{2}\thanks{\texttt{larsen@cs.au.dk}. Supported by
          Independent Research Fund Denmark (DFF) Sapere Aude Research
        Leader grant No 9064-00068B and a Villum Young Investigator grant.}
	\and
	Mark Simkin\inst{2}\thanks{\texttt{simkin@cs.au.dk}. Supported by
          Independent Research Fund Denmark (DFF) Sapere Aude Research
        Leader grant No 9064-00068B.}
}
\institute{
	Ruhr University Bochum
	\and
	Aarhus University
}

\begin{document}
\pagestyle{plain}

\maketitle

\begin{abstract}
Property-preserving hash functions allow for compressing long inputs $x_0$ and $x_1$ into short hashes $h(x_0)$ and $h(x_1)$ in a manner that allows for computing a predicate $P(x_0, x_1)$ given only the two hash values without having access to the original data. 
Such hash functions are said to be adversarially robust if an adversary that gets to pick $x_0$ and $x_1$ after the hash function has been sampled, cannot find inputs for which the predicate evaluated on the hash values outputs the incorrect result.

In this work we construct robust property-preserving hash functions for the hamming-distance predicate which distinguishes inputs with a hamming distance at least some threshold $t$ from those with distance less than $t$. The security of the construction is based on standard lattice hardness assumptions.

Our construction has several advantages over the best known previous construction by Fleischhacker and Simkin (Eurocrypt 2021).
Our construction relies on a single well-studied hardness assumption from lattice cryptography whereas the previous work relied on a newly introduced family of computational hardness assumptions.
In terms of computational effort, our construction only requires a small number of modular additions per input bit, whereas the work of Fleischhacker and Simkin required several exponentiations per bit as well as the interpolation and evaluation of high-degree polynomials over large fields.
An additional benefit of our construction is that the description of the hash function can be compressed to $\lambda$ bits assuming a random oracle.
Previous work has descriptions of length $\mathcal{O}(\ell \lambda)$ bits for input bit-length $\ell$, which has a secret structure and thus cannot be compressed.

We prove a lower bound on the output size of any property-preserving hash function for the hamming distance predicate.
The bound shows that the size of our hash value is not far from optimal.

\end{abstract}


\section{Introduction}\label{sec:intro}

Efficient algorithms that compress large amounts of data into small digests that preserve certain properties of the original input data are ubiquitous in computer science and hardly need an introduction. 
Sketching algorithms~\cite{STOC:AloMatSze96}, approximate membership data structures~\cite{ACM:Bloom70}, locality-sensitive hash functions~\cite{STOC:IndMot98}, streaming algorithms~\cite{SODA:Muthukrishnan03}, and compressed sensing~\cite{IEEE:Donoho06} are only a few among many examples.

Commonly, these algorithms are studied in benign settings where no adversarial parties are present. 
More concretely, these randomized algorithms usually state their (probabilistic) correctness guarantees by quantifying over all inputs and arguing that with high probability over the chosen random coins, the algorithm will behave as it should.
Importantly, the inputs to the algorithm are considered to be \emph{independent} of the random coins used.

In real world scenarios, however, the assumption of a benign environment may not be justified and an adversary may be incentivized to manipulate a given algorithm into outputting incorrect results by providing malicious inputs.
Adversaries that choose their inputs adaptively \emph{after} the random coins of the algorithm have been sampled, were previously studied in the context of sketching and streaming algorithms~\cite{STOC:MirNaoSeg08,STOC:HarWoo13,C:NaoYog15,CCS:ClaPatShr19,ITCS:BoyLaVVai19,PODS:BJWY20,PODS:BY20,EC:FleSim21}. 
These works show that algorithms which work well in benign environments are not guaranteed to work well in the presence of adaptive malicious inputs and several algorithms with security guarantees against malicious inputs were proposed.

The focus of this work are adversarially robust property-preserving hash (PPH) functions, recently introduced by Boyle, LaVigne, and Vaikuntanathan~\cite{ITCS:BoyLaVVai19}, which allow for compressing long inputs $x_0$ and $x_1$ into short hashes $h(x_0)$ and $h(x_1)$ in a manner that allows for evaluating a predicate $P(x_0, x_1)$ given only the two hash values without having access to the original data. 
A bit more concretely, a PPH function for a predicate $P : X \times X \to \bin$ is composed of a deterministic compression function $h : X \to Y$ and an evaluation algorithm $\eval : Y \times Y \to \bin$.
Such a pair of functions is said to be adversarially robust if no computationally bounded adversary $\adv$, who is given a random $(h, \eval)$ from an appropriate family, can find inputs $x_0$ and $x_1$, such that $P(x_0, x_1) \neq \eval(h(x_0), h(x_1))$.

BLV constructed PPH functions that compress inputs by a constant factor for the gap hamming predicate, which distinguishes inputs with very small hamming distance from those with a large distance\footnote{We do not care about the exact size of their gap, since we will focus on a strictly stronger predicate in this work.}.
For inputs that have neither a very small or very large distance, their construction provided no guarantees.

Subsequently Fleischhacker and Simkin~\cite{EC:FleSim21} constructed PPH functions for the exact hamming distance predicate, which distinguishes inputs with distance at least $t$ from those with distance less than $t$.
Their construction compresses arbitrarily long inputs into hash values of size $\bigO{t \secpar}$, where $\secpar$ is the computational security parameter.
Unfortunately, their construction is based on a new family of computational assumptions, which is introduced in their work, meaning that the security of their result is not well understood.
From a computational efficiency point of view, their construction is rather expensive.
It requires \bigO{\inlen} exponentiations for hashing a single $\ell$-bit long input and evaluating the predicate on the hashes requires  interpolating and evaluating high-degree polynomials over large fields.

\subsection{Our Contribution}

In this work we present a new approach for constructing PPH functions for the exact hamming distance predicate, which improves upon the result of Fleischhacker and Simkin in several ways.

The security of our construction relies on a well-studied hardness assumption from the domain of lattice-based cryptography.
Both hashing an input and evaluating a predicate on hash values only involves fast operations, such as modular additions, xor, and evaluating a few $t$-wise independent hash functions.
The size of our hash values is $\mathcal{\tilde O}(\secpar^2 t)$ bits.
We present a lower bound of roughly $\Omega(t)$ on the size of the hash value of any PPH function for the exact hamming distance predicate, showing that our result is not far from optimal.

Our hash functions can be described by a uniformly random bit string of sufficient length.
This means that, assuming a random oracle, these descriptions can compressed into $\secpar$ bits by replacing it with a short seed.
This compression is not applicable to the work of Fleischhacker and Simkin, since their hash function descriptions are bit strings with a secret structure that is only known to the sampling algorithm.

\subsection{Technical Overview}

Let $x_0$ and $x_1$ be two $\inlen$-bit strings, which we would like to compress using a hash function $h$ in a manner that allows us to use $h(x_0)$ and $h(x_1)$ to check whether $\dist(x_0, x_1) < t$, where $d$ is the hamming distance and $t$ is some threshold.
We start with a simple observation from the work of Fleischhacker and Simkin~\cite{EC:FleSim21}.
We can encode bit strings $x = x_1 x_2 \dots x_\inlen$ into sets $X= \{2i-x_{i}\mid i = 1, \dots, \inlen\}$ and for $x_0, x_1 \in \bin^\inlen$ we have that $\dist(x_0, x_1) < t$, if and only if $\abs{X_0 \ssdiff X_1} < 2t$.
Thus, from now on we can focus on hashing sets and constructing a property-preserving hash function for the symmetric set difference, which turns out to be an easier task.

Conceptually, our construction is inspired by Invertible Bloom Lookup Tables (IBLTs), which were introduced by Goodrich and Mitzenmacher~\cite{ALL:GooMit11}.
This data structure allows one to encode a set into an $\mathcal{\tilde O}(t)$ sketch with the following properties:
Two sketches can be subtracted from each other, resulting in a new sketch that corresponds to an encoding of the symmetric set difference of the original sets.
A sketch that contains at most $\bigO{t}$ many set elements can be fully decoded with high probability.

Given this data structure, one could attempt the following construction of a PPH function for the symmetric set difference predicate.
Given an input set, encode it as an IBLT.
To evaluate the symmetric set difference predicate on two hash values, subtract the two given IBLTs and attempt to decode the resulting data structure.
If decoding succeeds, then count the number of decoded elements and check, whether it's more or less than $2t$.
If decoding fails, then conclude that the symmetric set difference is too large.
The main issue with this construction is that IBLTs do not provide any correctness guarantees for inputs that are chosen adversarially.
Thus, the main contribution of this work is to construct a robust set encoding similar to IBLTs that remains secure in the presence of an adversary.

Our robust set encoding is comprised of ``random'' functions $r_i : \bin^* \to \{1, \dots, 2t\}$ for $i = 1, \dots, \statpar$ and a ``special'' collision-resistant hash function $A$.
To encode a set $X$, we generate an initially empty $\statpar \times 2t$ matrix $H$.
Each element $x \in X$ is then inserted by adding $A(x)$ in each row $i$ to column $r_i(x)$ in $H$, i.e., $H[i,r_i(x)] = H[i,r_i(x)] + A(x)$ for $i = 1, \dots, \statpar$.
To subtract two encodings, we simply subtract the two matrices entry-wise.
To decode a matrix back into a set, we repeatedly look for entries in $H$ that contain a single hash value $A(x)$, i.e., for cells $i, j$ with $\abs{H[i,j]} = A(x)$ for some $x$, and peel them away.
That is, whenever we find such an entry, we find $x$ corresponding to $A(x)$ and then remove $x$ from all positions, where it was originally inserted in $H$. 
Then we repeat the process until the matrix $H$ is empty or until the process gets stuck, because no cell contains a single set element by itself.

To prove security of our construction, we will show two things.
First, we will show that no adversary can find a pair of sets that have a small symmetric set difference, where the peeling process will get stuck. Actually, we will show something stronger, namely that such pairs do not exist with overwhelming probability over the random choices of $r_1, \dots, r_\statpar$.
Secondly, we will need to show that no (computationally bounded) adversary can find inputs, which decode incorrectly.
In particular, we will have to argue that the peeling process never decodes an element that was not actually encoded, i.e., that the sum of several hash values in some cell $H[i,j]$ never looks like $A(x)$ for some single set element $x$.
To argue that such a bad sum of hash values does not exist, one would need to pick the output length of $A$ too big in the sense that our resulting PPH function would not be compressing.
Instead, we will show that for an appropriate choice of $A$ these sums may exist, but finding them is hard and can be reduced to the computational hardness of solving the Short Integer Solution Problem~\cite{STOC:Ajtai96}, a well-studied assumption from lattice-based cryptography.


\section{Preliminaries}\label{sec:prelim}

This section introduces notation, some basic definitions and lemmas that we will use throughout this work.
We denote by $\secpar\in\NN$ the security parameter and by $\poly$ any function that is bounded by a polynomial in $\secpar$.
A function $f$ in $\secpar$ is negligible, if for every $c \in \NN$, there exists some $N\in\NN$, such that for all $\secpar>N$ it holds that $f(\secpar) < 1/\secpar^c$.
We denote by $\negl$ any negligible function.
An algorithm is PPT if it is modeled by a probabilistic Turing machine with a running time bounded by $\poly$.

We write $e_i$ to denote the $i$-th canonical unit vector, i.e. the vector of zeroes with a one in position $i$, and assume that the dimension of the vector is known from the context.
For a row vector $v$, we write $v^\intercal$ to denote its transpose.
Let $n\in\NN$, we denote by $[n]$ the set $\{1,\dots,n\}$.
Let $X,Y$ be sets, we denote by $\abs{X}$ the size of $X$ and by $X \ssdiff Y$ the symmetric set difference of $X$ and $Y$, i.e., $X \ssdiff Y = (X \cup Y)\setminus (X\cap Y) = (X \setminus Y)\cup (Y\setminus X)$.
We write $x\gets X$ to denote the process of sampling an element of $X$ uniformly at random.
For $x, y\in \bin^n$, we write $\wei(x)$ to denote the Hamming weight of $x$ and we write $\dist(x, y)$ to denote the Hamming distance between $x$ and $y$, i.e., $\dist(x,y) = \wei(x \oplus y)$.
We write $x_i$ to denote the $i$-th bit of $x$.

\subsection{Property-Preserving Hash Functions}\label{sec:pph-def}

The following definition of property-preserving hash functions is taken almost verbatim from~\cite{ITCS:BoyLaVVai19}.
In this work, we consider the strongest of several different security notions that were proposed in~\cite{ITCS:BoyLaVVai19}.

\begin{definition}[Property-Preserving Hash]
For a $\secpar\in \NN$ an $\eta$-compressing property-preserving hash function family $\hfam_\secpar = \{h : X \to Y\}$ for a two-input predicate
requires the following three efficiently computable algorithms:
\begin{description}
\item[$\sample(1^\secpar) \to h$] is an efficient randomized algorithm that samples an efficiently computable random hash function from $\hfam$ with security parameter $\secpar$.
\item[$\hash(h,x) \to y$] is an efficient deterministic algorithm that evaluates the hash function $h$ on $x$. 
\item[$\eval(h, y_0,y_1) \to \bin$:] is an efficient deterministic algorithm that on input $h$, and $y_0,y_1 \in Y$ outputs a single bit.
\end{description}

We require that $\hfam$ must be compressing, meaning that $\log \vert Y \vert \leq \eta \log \vert X \vert$ for $0 < \eta < 1$.
\end{definition}

For notational convenience we write $h(x)$ for $\hash(h,x)$.

\begin{definition}[Direct-Access Robustness]
A family of PPH functions $\hfam = \{h : X \to Y\}$ for a two-input predicate $P : X\times X \to \bin$ is a family of direct-access robust PPH functions if, for any PPT adversary $\adv$ it holds that,
\[
\Pr\mleft[\begin{aligned} &h \gets \sample(\secparam);\\&(x_0, x_1) \gets \adv(h)\end{aligned} : \eval(h, h(x_0), h(x_1)) \neq P(x_0,x_1)\mright] \leq \negl,
\]
where the probability is taken over the internal random coins of $\sample$ and $\adv$.
\end{definition}

\subsubsection{Two-Input Predicates.}
We define the following two-input predicates, which will be the main focus of this work.

\begin{definition}[Hamming Predicate]\label{def:ham}
For $x, y \in \bin^n$ and $t >0$, the 
two-input predicate is defined as
\[
\ham^t(x, y) = 
 \begin{cases}
   1 &\text{if } \dist(x,y) \geq t \\
   0 & \text{Otherwise}\\ 
 \end{cases}
\]
\end{definition}

\subsection{Lattices}\label{sec:lattices}

In the following we recall some lattice hardness assumptions and the relationships between them.
We start by revisiting one of the most well-studied computational problems.

\begin{definition}[Shortest Independent Vector Problem]
For an approximation factor of $\approxfactor := \approxfactor(\latticedim) \geq 1$, the $(\latticedim, \approxfactor)$-SIVP is defined as follows: Given a lattice $\lattice \subset \RR^\latticedim$, output $\latticedim$ linearly independent lattice vectors, which have all euclidean length at most $\approxfactor \cdot \mindist$, where $\mindist$ is the minimum possible.
\end{definition}

Starting with the celebrated work of Lenstra, Lenstra, and Lov{\'a}sz~\cite{LLL82}, a long line of research works~\cite{STOC:ADRS15,SOSA:AggSte18,C:ALNS20} has been dedicated to finding fast algorithms for solving the exact and approximate shortest independent vector problem.
All existing algorithms for finding any $\mathsf{poly}(n)$-approximation run in time $2^{\Omega(n)}$ and it is believed that one can not do better asymptotically as is captured in the following assumption.

\begin{assumption}\label{assumptionSIVP}
For large enough $n$, there exists no $2^{o(n)}$-time algorithm for solving the $(n, \approxfactor)$-SIVP with $\approxfactor = \mathsf{poly}(n)$.
\end{assumption}

A different computationally hard problem that has been studied extensively is the short integer solution problem.

\begin{definition}[Short Integer Solution Problem]
For parameters $\latticedim, \samplesize, \domainsize, \solbound_2, \solbound_\infty \in \NN$, the $(\latticedim, \samplesize, \domainsize, \solbound_2, \solbound_\infty)$-SIS problem is defined as follows: Given a uniformly random matrix $A \in \ZZ_q^{\latticedim \times \samplesize}$, find $s \in \ZZ^m$ with $\norm{s}_2 \leq \solbound_2$ and $\norm{s}_\infty \leq \solbound_\infty$, such that $As^\intercal = 0$.
\end{definition}

It was shown by Micciancio and Peikert that the difficulty of solving the SIS problem fast on average is related to the difficulty of solving the SIVP in the worst-case.

\begin{theorem}[Worst-Case to Average-Case Reduction for SIS~\cite{C:MicPei13}]
Let $\latticedim$, $\samplesize := \samplesize(\latticedim)$, and $\solbound_2 \geq \solbound_\infty \geq 1$ be integers.
Let $\domainsize \geq \solbound_2 \cdot \latticedim^\delta$ for some constant $\delta > 0$.
Solving the $(\latticedim, \samplesize, \domainsize, \solbound_2, \solbound_\infty)$-SIS problem on average with non-negligible probability in $n$ is at least as hard as solving the $(\latticedim, \approxfactor)$-SIVP in the worst-case to within $\approxfactor = \max(1, \solbound_2 \cdot \solbound_\infty/\domainsize) \cdot \tildebigO{\solbound_2 \sqrt{\latticedim}}$.
\end{theorem}

Combining the above result with \autoref{assumptionSIVP}, we get the following corollary.

\begin{corollary}
Let $n \in \Theta(\secpar)$ and $m = \poly$ be integers, let $\solbound_\infty = 2$, and let $\solbound_2 = \sqrt{m + \nu}$ for some constant $\nu$.
Let $q > \solbound_2 \cdot n^\delta$ for some constant $\delta > 0$.
If~\autoref{assumptionSIVP} holds, then for large enough $\secpar$, there exists no PPT adversary that solves the $(\latticedim, \samplesize, \domainsize, \solbound_2, \solbound_\infty)$-SIS problem with non-negligible (in $\secpar$) probability.
\end{corollary}


\section{Robust Set Encodings}\label{sec:datastructure}

In this section, we define our notion of robust set encodings.
The encoding transforms a possibly large set into a smaller sketch.
Given two sketches of sets with a small enough symmetric set difference, one should be able to decode the symmetric set difference.
The security of our encodings guarantees that no computationally bounded adversary can find a pair of sets where decoding either returns the incorrect result or fails even though the symmetric set difference between the encoded sets is small.

\begin{definition}[Robust Set Encodings]
A robust set encoding for a universe $U$ is comprised of the following algorithms:

\begin{description}
\item[$\sample(\secparam, t) \to \efunc$] is an efficient randomized algorithm that takes the security parameter $\secpar$ and threshold $t$ as input and returns an efficiently computable set encoding function $\efunc$ sampled from the family $\efam$.
\item[$\encode(\efunc, X) \to y$] is an efficient deterministic algorithm that takes set encoding function $\efunc$ and set $X \subset U$ as input and returns encoding $y$.
\item[$\decode(\efunc, y_0,y_1) \to X' / \bot$] is an efficient deterministic algorithm that takes set encoding function $\efunc$ and two set encodings $y_0,y_1$ as input and returns set $X'$ or $\bot$. 
\end{description}
We denote by $\enclen : \NN \times \NN \to \NN$ the function that describes the length of the encoding for a given security parameter $\secpar$ and threshold $t$.
For any two sets $X_0,X_1$ we use $X'\gets\decdiff(\efunc,X_0,X_1)$ as a shorthand notation for \[X'\gets\decode(\efunc,\encode(\efunc,X_0),\encode(\efunc,X_1)).\]

We say a set encoding is robust, if for any PPT adversary $\adv$ and any threshold $t \in \NN$ it holds that,
\[
\Pr\mleft[
  \begin{aligned} 
    &\efunc \gets \sample(\secparam, t);\\
    &(X_0, X_1) \gets \adv(\efunc,t);\\
    &X'\gets\decdiff(\efunc,X_0,X_1)
  \end{aligned} ~:~ 
  \begin{aligned}
    &X'\not\in\{ X_0 \ssdiff X_1,\bot\}\\ \lor~&(|X_0 \ssdiff X_1|< t\land X'= \bot)
  \end{aligned} \mright] \leq \negl,
\]
where the probability is taken over the random coins of the adversary $\adv$ and $\sample$.
\end{definition}

\subsection{Instantiation}
In this section we construct a set encoding for universe $[m]$ with $m =\poly$ by modifying Invertible Bloom Lookup Tables~\cite{ALL:GooMit11} to achieve security against adaptive malicious inputs. Since we are only encoding polynomially large sets and can leverage the cryptographic hardness of the SIS problem, we can get away with only maintaining a matrix of hash values in our sketch and we do not require the additional counter or value fields that were present in the original construction of Goodrich and Mitzenmacher.
\begin{figure}[t]
\centering
\begin{pcvstack}[center,boxed]
    \begin{pchstack}[center]
    \begin{pcvstack}
    \procedure{$\sample(\secparam, t)$}{
      \tikz[remember picture]{\coordinate (sample-for-head);}\pcforeach i \in [\statpar] \\
      \tikz[remember picture]{\coordinate (sample-for-foot);}\quad r_i \gets \rfam.\sample(\secparam) \\
      R := (r_1, \dots, r_\statpar) \\
      A \gets \ZZ_q^{n \times m} \\
      \pcreturn \efunc:= (R, A)
    }
    \pcvspace
      \procedure{$\encode(\efunc,X)$}{
      H := (0^n)^{\statpar \times 2t} \in (\ZZ^n_q)^{\statpar \times 2t}\\
      \tikz[remember picture]{\coordinate (encode-for1-head);}\pcforeach x \in X \\
      \quad \tikz[remember picture]{\coordinate (encode-for2-head);}\pcforeach i \in [\statpar] \\
      \tikz[remember picture]{\coordinate (encode-for1-foot);} \quad\tikz[remember picture]{\coordinate (encode-for2-foot);} \quad H[i, r_i(x)] := H[i, r_i(x)] + A e_x^\intercal\\
      \pcreturn y:= H
    }
    \pcvspace
    \end{pcvstack}
    \pchspace
    \begin{pcvstack}
  \procedure{$\decode(\efunc, H_0,H_1))$}{
    H := H_0-H_1\\
    X' := \emptyset \\
    \tikz[remember picture]{\coordinate (do);}\pcdo[]\\
    \quad Z := \mleft\{(x, w) \middle|\;
       \begin{pcmbox}\begin{aligned}
        \exists \>(i,j)\in[\statpar]\times[2t]\ldotp\, \\
        \> \land H[i,j]=w\\ \>\land w\in\{A e_{x}^\intercal,-A e_{x}^\intercal\}
      \end{aligned}
      \end{pcmbox}
    \mright\} \\
    \quad X' := X' \cup \{x\mid \exists w\ldotp\,(x,w)\in Z\}\\
    \tikz[remember picture]{\coordinate (while);}\quad H := \removeZ(\efunc,H,Z)\\
    \pcwhile Z\neq\emptyset\\
    \pcif H = (0^n)^{\statpar \times 2t}\\
    \quad \pcreturn X' \\
    \pcelse \\
    \quad \pcreturn \bot
  }
  \pcvspace
    \procedure{$\removeZ(\efunc, H,Z)$}{
    \tikz[remember picture]{\coordinate (remove-for1-head);}\pcforeach (x,w) \in Z\\
    \quad \tikz[remember picture]{\coordinate (remove-for2-head);}\pcforeach i \in [\statpar] \\
    \tikz[remember picture]{\coordinate (remove-for1-foot);}\quad\tikz[remember picture]{\coordinate (remove-for2-foot);} \quad H[i, r_i(x)] := H[i, r_i(x)] - w\\
    \pcreturn H
    }
  \end{pcvstack}
  \end{pchstack}
\end{pcvstack}
  \tikz[remember picture,overlay]{
    \draw ($(sample-for-head)+(1ex,-.3\baselineskip)$) -- ($(sample-for-foot)+(1ex,-.5\baselineskip)$);
    \draw ($(remove-for1-head)+(1ex,-.3\baselineskip)$) -- ($(remove-for1-foot)+(1ex,-.5\baselineskip)$);
    \draw ($(remove-for2-head)+(1ex,-.3\baselineskip)$) -- ($(remove-for2-foot)+(1ex,-.5\baselineskip)$);
    \draw ($(encode-for1-head)+(1ex,-.3\baselineskip)$) -- ($(encode-for1-foot)+(1ex,-.5\baselineskip)$);
    \draw ($(encode-for2-head)+(1ex,-.3\baselineskip)$) -- ($(encode-for2-foot)+(1ex,-.5\baselineskip)$);
    \draw ($(do)+(1ex,-.3\baselineskip)$) -- ($(while)+(1ex,-.5\baselineskip)$); 
  }
\caption{Construction of robust set encodings for universe $\ZZ_m$.}\label{fig:hse-construction}
\end{figure}
Refer to \autoref{fig:hse-construction} for a full description of the construction.
Before we prove that the construction is a robust set encoding we will first prove a few of its properties that will be useful in the following.

The following lemma effectively states that given the difference of two encodings there will always be a least one element that can can be peeled if the symmetric set difference is small enough.

  \begin{lemma}\label{lem:goodcell}
Let $\rfam$ be a family of $t$-wise independent hash functions $r : [m] \to [2t]$  and let $k \geq 2 \log_{3/e} m$.
With probability at least $1-2^{-\Omega(\statpar)}$, it simultaneously holds for all sets
$T \subseteq [m]$ with $0 < \abs{T} \leq t$ that there is at least one
$x \in T$ and one index $i \in [\statpar]$ such that $r_i(x) \neq r_i(y)$ for
all $y \in T \setminus \{x\}$. Here
the probability is taken over the random choice of the $r_i$'s.
\end{lemma}
\begin{proof}
Let $E$ denote the event that there is a set $T$ with $0 < |T| \leq t$
such that for all $x \in T$ and all $i \in [\statpar]$, there is a $y \in T
\setminus \{x\}$ with $r_i(x)=r_i(y)$. We show that $\Pr[E]$ is small.
The proof follows from a union bound over all $T \subseteq [m]$ with
$2 \leq \abs{T} \leq t$. So fix one such $T$. Let $E_T$ denote the event that there is no $i \in [\statpar]$ and $x \in T$ such
that $r_i(x) \neq r_i(y)$ for
all $y \in T \setminus \{x\}$. Then by a union bound, we have \[\Pr[E]
\leq \Pr\Bigl[\smashoperator{\bigcup_{T\subseteq [m]}} E_T\Bigr] \leq \smashoperator{\sum_{T\subseteq [m]}} \Pr[E_T].\]

To bound $\Pr[E_T]$, notice that conditioned on $E_T$, the number of distinct hash values
$|\{r_i(x) \mid x \in T\}|$ for the $i$th hash function is at most $|\abs{T}/2$, as every hash value is
\emph{hit} by either $0$ or at least $2$ elements from $T$. Now define
an event $E_{T,S}$ for every $\statpar$-tuple $S=(S_1,\dots,S_\statpar)$ where $S_i$
is a subset of $|T|/2$ values in $[\statpar]$. The event $E_{T,S}$ occurs
if $r_i(x) \in S_i$ for every $x \in T$ and every $i \in [\statpar]$. If $E_T$ happens then at
least one event $E_{T,S}$ happens. Thus \[\Pr[E_T] \leq \Pr\Bigl[\bigcup_{S}
E_{T,S}] \leq \sum_S \Pr[E_{T,S}].\]

To bound $\Pr[E_{T,S}]$, notice
that by $t$-wise independence, the values $r_i(x)$ are independent and
fall in $S_i$ with
probability exactly $|T|/(2 \cdot 2t)$. Since this must happen for
every $i$ and every $x \in T$, we get that $\Pr[E_{T,S}] \leq (|T|/(4t))^{|T|\statpar}$ and $\Pr[E_T]
\leq \binom{2t}{|T|/2}^\statpar (|T|/(4t))^{|T|\statpar}$. A union bound over all $T$ gives
us $\Pr[E] \leq \sum_{j=2}^t \binom{m}{j} \binom{2t}{j/2}^\statpar
(j/(4t))^{j\statpar}$. Using the bound $\binom{n}{\statpar} \leq (en/\statpar)^\statpar$ for all $0
\leq \statpar \leq n$ and the bound $\binom{m}{j} \leq m^j$, we finally conclude:
\begin{eqnarray*}
  \Pr[E] &\leq& \sum_{j=2}^t \binom{m}{j} \binom{2t}{j/2}^\statpar
                (j/(4t))^{j\statpar} \\
  &\leq& \sum_{j=2}^t m^j (4et/j)^{j\statpar/2}
         (j/(4t))^{j\statpar} \\
  &=&\sum_{j=2}^t m^j(e/3)^{j\statpar/2} (3j/(4t))^{j\statpar/2}
\end{eqnarray*}
For $\statpar \geq 2 \log_{3/e} m$ we have
$(e/3)^{\statpar/2} \leq 1/m$. The above is thus bounded by
\begin{eqnarray*}
  \Pr[E] &\leq& \sum_{j=2}^t (3j/(4t))^{j\statpar/2} \\
  &\leq& \sum_{j=2}^t (3/4)^{j\statpar/2}
\end{eqnarray*}
For any $\statpar \geq 2$, the terms in this sum go down by a factor at least
$4/3$ and thus is bounded by $2^{-\Omega(\statpar)}$.\qed
\end{proof}

In the next lemma we show that correctly peeling one layer of elements during decoding leads to a state that is equivalent to never having inserted those elements in the first place.

\begin{lemma}\label{lem:correctdecodinggivesgoodencoding}
For any security parameter $\secpar$, any threshold $t$, any encoding function $\efunc\gets\sample(\secparam,t)$, any pair of subsets $X_0,X_1\subseteq [m]$ and any set 
\[
Z\subseteq\{(x,Ae_x) \mid x\in X_0\setminus X_1\} \cup \{(x,-Ae_x) \mid x\in X_1\setminus X_0\}
\]
and $X:=\{x\mid \exists w\ldotp\,(x,w)\in Z\}$ it holds that
\[\removeZ(\encode(\efunc,X_0)-\encode(\efunc,X_1),Z)\\ = \encode(\efunc,X_0\setminus X)-\encode(\efunc,X_1\setminus X).\]
\end{lemma}
\begin{proof} 
  Let $H_b := \encode(\efunc,X_b)$, $H'_b := \encode(\efunc,X'_b\setminus X)$ and $H := \removeZ(\efunc,H_0-H_1,Z)$
  For any $(i,j) \in [\statpar]\times[2t]$, let $S_{i,j} = \{x\in [m] \mid r_i(x)=j\}$.
  Then for each $(i,j) \in [\statpar]\times[2t]$ we have 
  \begin{align}
      \label{eq:defremove}H[i,j] = &H_0[i,j] - H_1[i,j] - \smashoperator{\sum_{x\in X\cap S_{i,j}}}Z(x)\\
       \label{eq:defencode}= &\smashoperator{\sum_{x\in X_0\cap S_{i,j}}} Ae_x^\intercal\; -\; \smashoperator{\sum_{x\in X_1\cap S_{i,j}}} Ae_x^\intercal\; - \;\smashoperator{\sum_{x\in X\cap S_{i,j}}}Z(x)\\
       \label{eq:subsetssd}= &\smashoperator{\sum_{x\in X_0\cap S_{i,j}}} Ae_x^\intercal\; - \;\smashoperator{\sum_{x\in X_1\cap S_{i,j}}} Ae_x^\intercal\; - \;\smashoperator{\sum_{x\in X\cap X_0\cap S_{i,j}}}Z(x)\; - \;\smashoperator{\sum_{x\in X\cap X_1\cap S_{i,j}}}Z(x)\\
       \label{eq:plusminusvalue}= &\smashoperator{\sum_{x\in X_0\cap S_{i,j}}} Ae_x^\intercal\; - \;\smashoperator{\sum_{x\in X_1\cap S_{i,j}}} Ae_x^\intercal\; - \;\smashoperator{\sum_{x\in X\cap X_0\cap S_{i,j}}}Ae_x^\intercal\;\; + \;\smashoperator{\sum_{x\in X\cap X_1\cap S_{i,j}}}Ae_x^\intercal\\
       \label{eq:subsetssd2}= &\smashoperator{\sum_{x\in (X_0\setminus X)\cap S_{i,j}}} Ae_x^\intercal\quad - \quad\smashoperator{\sum_{x\in (X_1\setminus X)\cap S_{i,j}}} Ae_x^\intercal\\
       \label{eq:defencode2}=&H_0'[i,j] - H_1'[i,j],
    \end{align}
    where we denote by $Z(x)$ the unique value $w$ such that $(x,w) \in Z$.
    Equations~\ref{eq:defremove} and \ref{eq:defencode} follow from the definitions of $\removeZ$ and $\encode$ respectively.
    Equations \ref{eq:subsetssd} and \ref{eq:subsetssd2} follow from the fact that $X$ is a subset of the symmetric set difference of $X_0$ and $X_1$.
    \autoref{eq:plusminusvalue} follows from the fact that $w=(-1)^{b} A e_{x}^\intercal$ iff $x\in X_b$.
    Finally, \autoref{eq:defencode2} follows again from the definition of $\encode$. \qed
\end{proof}

The following lemma essentially states that during the decoding process we will never peel an element that is not in the symmetric set difference \emph{and} all elements will be peeled correctly, i.e., the decoding algorithm correctly identifies whether an element is from $X_0$ or from $X_1$.

\begin{lemma}\label{lem:nobadZ}
  For an encoding function $\efunc\gets\sample(\secparam,t)$ and two sets $X_0,X_1$, let $Z_1,Z_2,\dots$ denote the sequence of sets peeled during the execution of $\decode(\efunc,\encode(\efunc,X_0),\encode(\efunc,X_0))$.
  If the $(n,m,\sqrt{m+3},2)$-SIS problem is hard, then for any PPT algorithm $\adv$, it holds that
  \[
    \Pr\mleft[
    \begin{aligned}
      \efunc:=\sample(\secparam,t);\\
      (X_0,X_1) \gets \adv(\efunc)
    \end{aligned}
      \exists c.\, Z_c \not\subseteq 
      \begin{aligned}
        &\{(x,Ae_x) \mid x\in X_0\setminus X_1\}\\ 
        \cup &\{(x,-Ae_x) \mid x\in X_1\setminus X_0\}
      \end{aligned}
    \mright]\leq\negl.
  \]
\end{lemma}
\begin{proof}
  Let $\adv$ be an arbitrary PPT algorithm with 
  \[
    \Pr\mleft[
    \begin{aligned}
      \efunc:=\sample(\secparam,t);\\
      (X_0,X_1) \gets \adv(\efunc)
    \end{aligned}
      \exists c.\, Z_c \not\subseteq 
      \begin{aligned}
        &\{(x,Ae_x) \mid x\in X_0\setminus X_1\}\\ 
        \cup &\{(x,-Ae_x) \mid x\in X_1\setminus X_0\}
      \end{aligned}
    \mright]=\epsilon(\secpar).
  \]
  We construct an algorithm $\bdv$ that solves $(n,m,\sqrt{m+3},2)$-SIS as follows.
  $\bdv$ receives as input a random matrix $A \in \ZZ_q^{n\times m}$, samples $r_i \gets \rfam$ for $i\in [\statpar]$ and invokes $\adv$ on $f=(A,(r_1,\dots,r_\statpar))$.
  Once $\adv$ outputs $X_0,X_1$, $\bdv$ runs $H_0 := \encode(\efunc,X_0)$ and $H_1 := \encode(\efunc,X_1)$ and then starts to execute $\decode(\efunc,H_0,H_1)$.
  Let $Z_c$ denote the set $Z$ in the $c$-th iteration of the main loop of $\decode$.
  In each iteration, if \[Z_c \not\subseteq \{(x,Ae_x^\intercal) \mid x\in X_0\setminus X_1\} \cup \{(x,-Ae_x^\intercal) \mid x\in X_1\setminus X_0\},\]
  then $\bdv$ stops the decoding process and proceeds as follows.
  
  Let $S_{i,j} = \{x\in [m] \mid r_i(x)=j\}$ and $X'_b := X_b\setminus (Z_1\cup\dots\cup Z_{c-1})$.
  By definition of $Z$, there must exists at least one element $(x,w)\in Z_c$, such that $H[i,j]=(-1)^{b}A e_{x}^\intercal$ and $x\not\in X_b\setminus X_{1-b}$ for some cell $(i,j)$ and some bit $b$.
  $\bdv$ identifies one such cell by exhaustive search and outputs the vector
  \[
    s := \smashoperator[r]{\sum_{y\in X'_0\cap S_{i,j}}}e_y\;-\;\smashoperator{\sum_{y\in X'_1\cap S_{i,j}}}e_y - (-1)^be_x.
  \]
  
  If the decoding procedure terminates without such a $Z_c$ occurring, $\bdv$ outputs $\bot$.
  
  To analyze the success probability of $\bdv$, consider that by \autoref{lem:correctdecodinggivesgoodencoding} and since $Z_c$ is the \emph{first} set in which an element as specified above exists, we have that $H = \encode(\efunc,X_0')-\encode(\efunc,X_1')$, i.e.
  \[
    (-1)^bAe_x^\intercal = H[i,j] = \sum_{y\in X'_0\cap S_{i,j}}Ae_y^\intercal\;-\;\smashoperator{\sum_{y\in X'_1\cap S_{i,j}}}Ae_y^\intercal
  \]
  Thus, whenever $\bdv$ outputs a vector $s$, it holds that $As^\intercal=0$.
  Furthermore, this vector consists of the sum of at most $m$ unique canonical unit vectors and one additional canonical unit vector. This inplies that $\norm{s}_2 \leq \sqrt{m+3}$ and $\norm{s}_\infty \leq 2$.
  We can conclude that $\bdv$ solves $(n,m,\sqrt{m+3},2)$-SIS, with probability $\epsilon(\secpar)$.
  Since $(n,m,\sqrt{m+3},2)$-SIS is assumed to be hard, $\epsilon(\secpar)$ must be negligible. \qed
\end{proof}

The following lemma states that with overwhelming probability the decoding process will output either $\bot$ or a subset of the symmetric set difference, even for maliciously chosen sets $X_0,X_1$.

\begin{lemma}\label{lem:noforeignelements}
If the $(n,m,\sqrt{m+3},2)$-SIS problem is hard, then for any PPT adversary $\adv$ it holds that
\[
  \Pr\mleft[
    \begin{aligned}
      \efunc:=\sample(\secparam,t);\\
      (X_0,X_1) \gets \adv(\efunc);\\
      X' := \decdiff(\efunc, X_0, X_1)
    \end{aligned}~:~ 
      X' \neq \bot
      \land\, X' \not\subseteq X_0 \ssdiff X_1
    \mright]
    \leq\negl
\]
\end{lemma}
\begin{proof}
Let $Z_1,Z_2,\dots$ denote the sequence of sets peeled during the execution of \[\decode(\efunc,\encode(\efunc,X_0),\encode(\efunc,X_0)).\]
If an algorithm outputs $X_0,X_1$, such that $X' \not\subseteq X_0 \ssdiff X_1$, there must exist an $x\in X'$ such that \[x'\not\in X_0 \ssdiff X_1 = (X_0\setminus X_1) \cup (X_1\setminus X_0).\]
Since $X' := \{x\mid \exists w.\, (x,w) \in Z_1 \cup\dots \}$, this can only happen with negligible probability by \autoref{lem:nobadZ}. \qed
\end{proof}

The following lemma states that with overwhelming probability the decoding process will never output a \emph{strict} subset of the symmetric set difference, even for maliciously chosen sets $X_0,X_1$.

\begin{lemma}\label{lem:nostrictsubset}
If the $(n,m,\sqrt{m+3},2)$-SIS problem is hard, then for any PPT adversary $\adv$ it holds that 
\[
  \Pr\mleft[
    \begin{aligned}
      \efunc:=\sample(\secparam,t);\\
      (X_0,X_1) \gets \adv(\efunc);\\
      X' := \decdiff(\efunc, X_0, X_1)
    \end{aligned}~:~ 
      X' \subsetneq X_0 \ssdiff X_1
    \mright]
    \leq\negl
\]
\end{lemma}
\begin{proof}
Let $\adv$ be a PPT an adversary for the above experiment.
We construct an adversary $\bdv$ against $(n, m, \sqrt{m+3}, 2)$-SIS as follows.
$\bdv$ is given matrix $A$, samples $r_i \gets \rfam$ for $i\in [\statpar]$ and invokes $\adv$ on $\efunc=(A,(r_1,\dots,r_\statpar))$.
Adversary $\adv$ returns $X_0$ and $X_1$ and $\bdv$ computes $X' := \decdiff(\efunc, X_0, X_1)$.
If $X' \subsetneq X_0 \ssdiff X_1$, then $\bdv$ computes $X'_b = X_b \setminus X'$ for $b\in\{0,1\}$ and finds an index $i, j$ such that there exists an $x \in X'_0 \ssdiff X'_1$ with $r_i(x) = j$.
$\bdv$ returns
\[
  s := \smashoperator[r]{\sum_{y\in X'_0\cap S_{i,j}}}e_y\;-\;\smashoperator{\sum_{y\in X'_1\cap S_{i,j}}}e_y.
\]
Since every canonical unit vector appears at most once in the sum above, it follows that $\norm{s}_2 \leq \sqrt{m}$ and $\norm{s}_\infty = 1$.

To analyze the probability that $As^\intercal=0$ we consider the following.
Let $H'$ be the value of the matrix $H$ when the decoding procedure terminates.
By \autoref{lem:nobadZ} and \autoref{lem:correctdecodinggivesgoodencoding} it holds with overwhelming probability that $H' = H'_0 - H'_1 = \encode(\efunc,X'_0)-\encode(\efunc,X'_0)$.
However, since the decoding terminates successfully, it must also hold that $H' = (0^n)^{\statpar \times 2t}$.
It follows that for all $i, j$, we have $H'_0[i,j] - H'_1[i,j] = 0$ and therefore $As = 0$ with overwhelming probability.
Since $(n,m,\sqrt{m+3},2)$-SIS is assumed to be hard the lemma follows. \qed
\end{proof}

By combining \autoref{lem:noforeignelements} and \autoref{lem:nostrictsubset} we obtain the following corollary stating that with overwhelming probability the decoding process will output \emph{either} the correct symmetric set difference \emph{or} the error symbol $\bot$.

\begin{corollary}\label{corr:exactdiff}
  If the $(n,m,\sqrt{m+3},2)$-SIS problem is hard, then for any PPT adversary $\adv$ it holds that
  \[
  \Pr\mleft[
    \begin{aligned}
      \efunc:=\sample(\secparam,t);\\
      (X_0,X_1) \gets \adv(\efunc);\\
      X' := \decdiff(\efunc, X_0, X_1)
    \end{aligned}~:~ 
      X'\not\in \{X_0 \ssdiff X_1,\bot\}
    \mright]
    \leq\negl
\]
\end{corollary}

The following lemma states that with overwhelming probability the decoding process will not output $\bot$ if the symmetric set difference is small.

\begin{lemma}\label{lem:worksforsmall}
  If the $(n,m,\sqrt{m+3},2)$-SIS problem is hard, then for any PPT adversary $\adv$ it holds that
  \[
\Pr\mleft[
  \begin{aligned} 
    &\efunc \gets \sample(\secparam, t);\\
    &(X_0, X_1) \gets \adv(\efunc,t);\\
    &X'\gets\decdiff(\efunc,X_0,X_1)
  \end{aligned} ~:~ 
    |X_0 \ssdiff X_1|< t\land X'= \bot\mright]\leq\negl
    \]
\end{lemma}

\begin{proof}
  Let $\adv$ be an arbitrary PPT algorithm.
  By \autoref{lem:nobadZ} and \autoref{lem:correctdecodinggivesgoodencoding} it holds that in each iteration $c$ we have $H = H_{c,0}-H_{c,1}$, where $H_{c,b} = \encode(\efunc,X_{c,0},X_{c,1})$ and $X_{c,b}=X_b\setminus\{x\mid\exists w.\; (x,w)\in Z_1\cup\dots\cup Z_{c-1}\}$.
Since it must hold that $|X_0 \ssdiff X_1|< t$ it in particular holds that $|X_{c,0} \ssdiff X_{c,1}| < t$ in each iteration.
By \autoref{lem:goodcell}, in each iteration where $X_{c,1} \ssdiff X_{c,2} \neq \emptyset$ it holds that $Z_c \neq \emptyset$ with overwhelming probability.
Therefore, the decoding process terminates after at most $t$ steps, with $X' = X_0 \ssdiff X_1$.
Since each peeling step was correct with overwhelming probability it must hold that $H=(0^n)^{\statpar\times 2t}$.\qed
\end{proof}

Given the above lemmas, we can now easily prove the following theorem.

\begin{theorem}
  Let $\rfam$ be a family of $t$-wise independent hash functions $r : [m] \to [2t]$  and let $\statpar \geq \max\{\secpar,2 \log_{3/e} m\}$. 
  Then the construction in \autoref{fig:hse-construction} is a robust set encoding for universe $[m]$ if the $(n=n(\secpar),m,\sqrt{m+3},3)$-SIS problem is hard.
\end{theorem}

\begin{proof}
Let $\adv$ be an arbitrary PPT algorithm, using \autoref{corr:exactdiff}, \autoref{lem:worksforsmall} and a simple union bound we can conclude that
\begin{align*}
&\Pr\mleft[
  \begin{aligned} 
    &\efunc \gets \sample(\secparam, t);\\
    &(X_0, X_1) \gets \adv(\efunc,t);\\
    &X'\gets\decdiff(\efunc,X_0,X_1)
  \end{aligned} ~:~ 
  \begin{aligned}
    &X'\not\in\{ X_0 \ssdiff X_1,\bot\}\\ \lor~&(|X_0 \ssdiff X_1|< t\land X'= \bot)
  \end{aligned} \mright]\\
  \leq~&\Pr\mleft[
  \begin{aligned} 
    &\efunc \gets \sample(\secparam, t);\\
    &(X_0, X_1) \gets \adv(\efunc,t);\\
    &X'\gets\decdiff(\efunc,X_0,X_1)
  \end{aligned} ~:~ 
    X'\not\in\{ X_0 \ssdiff X_1,\bot\}
  \mright]\\
  &+\Pr\mleft[
  \begin{aligned} 
    &\efunc \gets \sample(\secparam, t);\\
    &(X_0, X_1) \gets \adv(\efunc,t);\\
    &X'\gets\decdiff(\efunc,X_0,X_1)
  \end{aligned} ~:~ 
    |X_0 \ssdiff X_1|< t\land X'= \bot\mright]
   \\
   \leq~& \negl.
\end{align*}
\end{proof}

\begin{remark}
  Instantiated as specified, the construction has keys that consist of $\statpar$ many $t$-wise independent hash functions and a matrix $A\in\ZZ_q^{m\times n}$, leading to a key length of $\statpar t\cdot\log m + mn\cdot\log q$.
  Note that the entire key can be represented by a public uniformly random $\statpar t\cdot\log m + mn\cdot\log q$ bit string.
  Assuming the existence of a random oracle, this string can be replaced by a short $\secpar$ bit seed.
\end{remark}


\section{Construction}\label{sec:construction}
In this section we construct property-preserving hash functions for the exact hamming distance predicate based on robust set encodings.

\subsection{PPH for the Hamming Distance Predicate}\label{sec:pphssd}
\begin{figure}[t]
\begin{pcvstack}[center,boxed]
  \begin{pchstack}[center]
    \procedure{$\sample(\secparam)$}{
      \efunc\gets\efam.\sample(\secparam, 2t)\\
      \pcreturn h:= \efunc
    }
    \pchspace
    \procedure{$\hash(h,x)$}{
      X := \{2i-x_i\mid i\in[\inlen]\}\\
      y:=\efam.\encode(h,X)\\
      \pcreturn y
    }
    \pchspace
  \procedure{$\eval(h,y_0,y_1)$}{
    X' := \efam.\decode(h,y_0,y_1)\\
    \pcif X'=\bot \pcor \abs{X'} \geq 2t\\
    \quad \pcreturn 1\\
    \pcelse\\
    \quad \pcreturn 0
  }
  \end{pchstack}
\end{pcvstack}
\caption{A family of direct-access robust PPHs for the predicate $\ham^t$ over the domain $\bin^\inlen$ for any $\inlen\in\NN$.}\label{fig:exact-ham}
\end{figure}

\begin{theorem}\label{thm:pph-ssd}
Let $\inlen = \poly$ and $t\leq \inlen$.
Let $\efam$ be a robust set encoding for universe $[2\inlen]$ with encoding length $\enclen$.
Then, the construction in Figure~\ref{fig:exact-ham} is a $\enclen(\secpar,2t)/\inlen$-compressing direct-access robust property-preserving hash function family for the two-input predicate $\ham^t$ and domain $\bin^\inlen$.
\end{theorem}

\begin{proof}
  Let $\adv$ be an arbitrary PPT adversary against the direct-access robustness of $\hfam$.
  We construct an adversary $\bdv$ against the robustness of $\efam$ as follows.
  Upon input $e$, $\bdv$ invokes $\adv$ on input $h:=\efunc$.
  When $\adv$ outputs $x_0,x_1$, $\bdv$  outputs $X_0 := \{2i-x_{0,i} \mid i\in[\inlen]\}$ and $X_1 := \{2i-x_{1,i} \mid i\in[\inlen]\}$.
  We note that it holds that
  \begin{align}
    &\Pr\mleft[\begin{aligned} &h \gets \sample(\secparam);\\&(x_0, x_1) \gets \adv(h)\end{aligned} : \eval(h, h(x_0), h(x_1)) \neq \ham^t(x_0,x_1)\mright]\\
    \label{eq:defsamplehash}=&\Pr\mleft[\begin{aligned} &\efunc \gets \efam.\sample(\secparam, 2t);\\&(X_0, X_1) \gets \bdv(\efunc);\\
    &y_0:=\efam.\encode(\efunc,X_0);\\
    &y_1:=\efam.\encode(\efunc,X_1)
    \end{aligned} : \eval(\efunc, y_0, y_1) \neq \ham^t(x_0,x_1)\mright]\\
    \label{eq:defevalpred}=&\Pr\mleft[\begin{aligned} &\efunc \gets \efam.\sample(\secparam, 2t);\\&(X_0, X_1) \gets \bdv(\efunc);\\
    &X':=\decdiff(\efunc,X_0, X_1)
    \end{aligned} : 
    \begin{aligned}
      &(\dist(x_0,x_1)\geq t \land X'\neq\bot \land |X'| < 2t)\\
      \lor&(\dist(x_0,x_1) < t \land (X'=\bot \lor |X'| \geq 2t))
    \end{aligned}\mright]\\
    \label{eq:defsets}=&\Pr\mleft[\begin{aligned} &\efunc \gets \efam.\sample(\secparam, 2t);\\&(X_0, X_1) \gets \bdv(\efunc);\\
    &X':=\decdiff(\efunc,X_0, X_1)
    \end{aligned} : 
    \begin{aligned}
      &(\abs{X_0 \ssdiff X_1}\geq 2t \land X'\neq\bot \land |X'| < 2t)\\
      \lor&(\abs{X_0 \ssdiff X_1} < 2t \land (X'=\bot \lor |X'| \geq 2t))
    \end{aligned}\mright]\\
    \label{eq:splitclause}=&\Pr\mleft[\begin{aligned} &\efunc \gets \efam.\sample(\secparam, 2t);\\&(X_0, X_1) \gets \bdv(\efunc);\\
    &X':=\decdiff(\efunc,X_0, X_1)
    \end{aligned} : 
    \begin{aligned}
      &(\abs{X_0 \ssdiff X_1}\geq 2t \land X'\neq\bot \land |X'| < 2t)\\
      \lor& (\abs{X_0 \ssdiff X_1} < 2t \land X'\neq\bot \land |X'| \geq 2t)\\
      \lor&(\abs{X_0 \ssdiff X_1} < 2t \land X'=\bot)
    \end{aligned}\mright]\\
    \label{eq:combineclause}=&\Pr\mleft[\begin{aligned} &\efunc \gets \efam.\sample(\secparam, 2t);\\&(X_0, X_1) \gets \bdv(\efunc);\\
    &X':=\decdiff(\efunc,X_0, X_1)
    \end{aligned} : 
    \begin{aligned}
      &(X'\neq\bot \land \abs{X_0 \ssdiff X_1} \neq \abs{X'})\\
      \lor&(\abs{X_0 \ssdiff X_1} < 2t \land X'=\bot)
    \end{aligned}\mright]\\
    \label{eq:bsuccprob}\leq&\Pr\mleft[\begin{aligned} &\efunc \gets \efam.\sample(\secparam, 2t);\\&(X_0, X_1) \gets \bdv(\efunc);\\
    &X':=\decdiff(\efunc,X_0, X_1)
    \end{aligned} : 
    \begin{aligned}
      &X'\not\in\{X_0 \ssdiff X_1,\bot\}\\
      \lor&(\abs{X_0 \ssdiff X_1} < 2t \land X'=\bot)
    \end{aligned}\mright].\\
  \end{align}
\end{proof}
Here \autoref{eq:defsamplehash} follows from the definition of $\sample$ and $\hash$ and \autoref{eq:defevalpred} follows from the definition of $\eval$ as well as the exact hamming distance predicate.
\autoref{eq:defsets} follows from the definition of the sets $X_0,X_1$: for each position $i$ where the $x_{0,i}=x_{1,i}$, the sets share an element, whereas for every position where $x_{0,i}\neq x_{1,i}$, one of them contains the element $2i$ and the other $2i-1$, thus $d(x_0,x_1) = t \iff |X_0\ssdiff X_1| = 2t$.
Equations \ref{eq:splitclause} and \ref{eq:combineclause} follow by first splitting the bottom clause and then rewriting the top two clauses.

Finally, since $\efam$ is a robust set encoding it holds by assumption that the probability in \autoref{eq:bsuccprob} is negligible and the theorem thus follows.

\begin{corollary}
  Instantiating the construction from Figure~\ref{fig:exact-ham} using the robust set encoding from \autoref{sec:datastructure} with
  $\statpar=n=\secpar$ and $q=\sqrt{\secpar(2\inlen+3)}$ leads to a $\tfrac{2tkn\log q}{\inlen}=\tfrac{t\secpar^2\log (2\inlen+3)}{\inlen}$ compressing PPH for exact hamming distance.
\end{corollary}

\section{Lower Bound}
In this section, we show a lower bound on the output length of a PPH for exact Hamming distance. We prove the lower bound by reduction from indexing. In the indexing problem, there are two parameters $k$ and $m$. The first player Alice is given a string $x=(x_1,\dots,x_m) \in [k]^m$, while the second player Bob is given an integer $i \in [m]$. Alice sends a single message to Bob and Bob should output $x_i$. The following lower bound holds:
\begin{lemma}[\cite{MILTERSEN199837}]
  \label{lem:indexing}
  In any one-way protocol for indexing in the joint random source model with success probability at least $1-\delta > 3/(2k)$, Alice must send a message of size $\Omega((1-\delta)m \log k)$.
\end{lemma}
Here the joint random source model means that Alice and Bob have shared randomness that is drawn independently of their inputs. We prove the following lower bound:
\begin{theorem}
  Any PPH for the exact Hamming distance predicate on $\inlen$-bit strings with threshold $t$ and success probability at least $1-\delta$(This means that the robustness error is at most $\delta$.), must have an output length of $\Omega(t \log(\min\{\inlen/t, 1/\delta\}))$ bits.
\end{theorem}
\begin{proof}
  Assume that there is a PPH-family $\hfam$ for the predicate $\ham^t$ and input length $\inlen$ with $t \leq \inlen$, such that for any strings $x,y$
  \[
    \Pr\mleft[h \gets \sample(\secparam) : \eval(h, h(x), h(y)) \neq \ham^t(x,y)\mright] \leq \delta.\footnote{Note that this is a strictly weaker requirement than direct access robustness, since $x,y$ are not adversarially chosen depending on $h$.}
  \]
  Let $s$ denote the output length of $\hfam$. We then use $\hfam$ to solve indexing with parameters $k = \min\{\inlen/t, \delta^{-1}/2 \}$ and $m = t-1$. When Alice receives a string $x \in [k]^m$, she constructs a binary string $y$ consisting of $m$ chunks of $k$ bits. If $mk < \inlen$, she pads this string with $0$'s. Each chunk in $y$ has a single $1$ in position $x_i$ and $0$'s elsewhere. She then computes the hash value $h(y)$, where $h$ is sampled from $\hfam$ using joint randomness, and sends it to Bob, costing $s$ bits.

From his index $i \in [m]$, Bob constructs $k$ bit strings $z_1,\dots,z_k$ of length $\inlen$, such that $z_j$ has a $1$ in the position corresponding to the $j$'th position of the $i$'th chunk of $y$, and $0$ everywhere else. He then computes the hash values $h(z_1),\dots,h(z_k)$ (using the joint randomness to sample $h$) and runs $\eval(h,h(y),h(z_j))$. Bob outputs as his guess for $x_i$, an index $j$, such that $\eval(h,h(y),h(z_j)) = 0$. Notice that the Hamming distance between $z_j$ and $y$ is $m+1>t$ if $j \neq x_i$ and it is $m-1<t$ otherwise. Thus if all $k$ evaluations are correct, Bob succeeds in reporting $x_i$. By a union bound, Bob is correct with probability at least $1-k \delta \geq 1/2$. By Lemma~\ref{lem:indexing}, we conclude $s = \Omega(t \log(\min\{\inlen/t, 1/\delta\}))$.
\end{proof}

\bibliographystyle{alpha}
\bibliography{bib/abbrev0,bib/crypto,bib/extraref}

\begin{thebibliography}{BEJWY20}

\bibitem[ADRS15]{STOC:ADRS15}
Divesh Aggarwal, Daniel Dadush, Oded Regev, and Noah {Stephens-Davidowitz}.
\newblock Solving the shortest vector problem in {$2^n$} time using discrete
  {Gaussian} sampling: Extended abstract.
\newblock In Rocco~A. Servedio and Ronitt Rubinfeld, editors, {\em 47th Annual
  {ACM} Symposium on Theory of Computing}, pages 733--742, Portland, OR, USA,
  June~14--17, 2015. {ACM} Press.

\bibitem[Ajt96]{STOC:Ajtai96}
Mikl{\'o}s Ajtai.
\newblock Generating hard instances of lattice problems (extended abstract).
\newblock In {\em 28th Annual {ACM} Symposium on Theory of Computing}, pages
  99--108, Philadephia, PA, USA, May~22--24, 1996. {ACM} Press.

\bibitem[ALNS20]{C:ALNS20}
Divesh Aggarwal, Jianwei Li, Phong~Q. Nguyen, and Noah {Stephens-Davidowitz}.
\newblock Slide reduction, revisited - filling the gaps in {SVP} approximation.
\newblock In Daniele Micciancio and Thomas Ristenpart, editors, {\em Advances
  in Cryptology -- {CRYPTO}~2020, Part~II}, volume 12171 of {\em Lecture Notes
  in Computer Science}, pages 274--295, Santa Barbara, CA, USA, August~17--21,
  2020. Springer, Heidelberg, Germany.

\bibitem[AMS96]{STOC:AloMatSze96}
Noga Alon, Yossi Matias, and Mario Szegedy.
\newblock The space complexity of approximating the frequency moments.
\newblock In {\em 28th Annual {ACM} Symposium on Theory of Computing}, pages
  20--29, Philadephia, PA, USA, May~22--24, 1996. {ACM} Press.

\bibitem[ASD18]{SOSA:AggSte18}
Divesh Aggarwal and Noah Stephens-Davidowitz.
\newblock {Just Take the Average! An Embarrassingly Simple $2^n$-Time Algorithm
  for SVP (and CVP)}.
\newblock In Raimund Seidel, editor, {\em 1st Symposium on Simplicity in
  Algorithms (SOSA 2018)}, volume~61 of {\em OpenAccess Series in Informatics
  (OASIcs)}, pages 12:1--12:19, Dagstuhl, Germany, 2018. Schloss
  Dagstuhl--Leibniz-Zentrum fuer Informatik.

\bibitem[BEJWY20]{PODS:BJWY20}
Omri Ben-Eliezer, Rajesh Jayaram, David~P Woodruff, and Eylon Yogev.
\newblock A framework for adversarially robust streaming algorithms.
\newblock In {\em Proceedings of the 39th ACM SIGMOD-SIGACT-SIGAI Symposium on
  Principles of Database Systems}, pages 63--80, 2020.

\bibitem[BEY20]{PODS:BY20}
Omri Ben-Eliezer and Eylon Yogev.
\newblock The adversarial robustness of sampling.
\newblock In {\em Proceedings of the 39th ACM SIGMOD-SIGACT-SIGAI Symposium on
  Principles of Database Systems}, pages 49--62, 2020.

\bibitem[Blo70]{ACM:Bloom70}
Burton~H Bloom.
\newblock Space/time trade-offs in hash coding with allowable errors.
\newblock {\em Communications of the ACM}, 13(7):422--426, 1970.

\bibitem[BLV19]{ITCS:BoyLaVVai19}
Elette Boyle, Rio LaVigne, and Vinod Vaikuntanathan.
\newblock Adversarially robust property-preserving hash functions.
\newblock In Avrim Blum, editor, {\em ITCS 2019: 10th Innovations in
  Theoretical Computer Science Conference}, volume 124, pages 16:1--16:20, San
  Diego, CA, USA, January~10--12, 2019. {LIPIcs}.

\bibitem[CPS19]{CCS:ClaPatShr19}
David Clayton, Christopher Patton, and Thomas Shrimpton.
\newblock Probabilistic data structures in adversarial environments.
\newblock In Lorenzo Cavallaro, Johannes Kinder, XiaoFeng Wang, and Jonathan
  Katz, editors, {\em ACM CCS 2019: 26th Conference on Computer and
  Communications Security}, pages 1317--1334. {ACM} Press, November~11--15,
  2019.

\bibitem[Don06]{IEEE:Donoho06}
David~L Donoho.
\newblock Compressed sensing.
\newblock {\em IEEE Transactions on information theory}, 52(4):1289--1306,
  2006.

\bibitem[FS21]{EC:FleSim21}
Nils Fleischhacker and Mark Simkin.
\newblock Robust property-preserving hash functions for hamming distance and
  more.
\newblock In {\em Annual International Conference on the Theory and
  Applications of Cryptographic Techniques}. Springer, 2021.

\bibitem[GM11]{ALL:GooMit11}
Michael~T Goodrich and Michael Mitzenmacher.
\newblock Invertible bloom lookup tables.
\newblock In {\em 2011 49th Annual Allerton Conference on Communication,
  Control, and Computing (Allerton)}, pages 792--799. IEEE, 2011.

\bibitem[HW13]{STOC:HarWoo13}
Moritz Hardt and David~P. Woodruff.
\newblock How robust are linear sketches to adaptive inputs?
\newblock In Dan Boneh, Tim Roughgarden, and Joan Feigenbaum, editors, {\em
  45th Annual {ACM} Symposium on Theory of Computing}, pages 121--130, Palo
  Alto, CA, USA, June~1--4, 2013. {ACM} Press.

\bibitem[IM98]{STOC:IndMot98}
Piotr Indyk and Rajeev Motwani.
\newblock Approximate nearest neighbors: Towards removing the curse of
  dimensionality.
\newblock In {\em 30th Annual {ACM} Symposium on Theory of Computing}, pages
  604--613, Dallas, TX, USA, May~23--26, 1998. {ACM} Press.

\bibitem[LLL82]{LLL82}
Arjen~K Lenstra, Hendrik~Willem Lenstra, and L{\'a}szl{\'o} Lov{\'a}sz.
\newblock Factoring polynomials with rational coefficients.
\newblock {\em Mathematische annalen}, 261:515--534, 1982.

\bibitem[MNS08]{STOC:MirNaoSeg08}
Ilya Mironov, Moni Naor, and Gil Segev.
\newblock Sketching in adversarial environments.
\newblock In Richard~E. Ladner and Cynthia Dwork, editors, {\em 40th Annual
  {ACM} Symposium on Theory of Computing}, pages 651--660, Victoria, BC,
  Canada, May~17--20, 2008. {ACM} Press.

\bibitem[MNSW98]{MILTERSEN199837}
Peter~Bro Miltersen, Noam Nisan, Shmuel Safra, and Avi Wigderson.
\newblock On data structures and asymmetric communication complexity.
\newblock {\em Journal of Computer and System Sciences}, 57(1):37--49, 1998.

\bibitem[MP13]{C:MicPei13}
Daniele Micciancio and Chris Peikert.
\newblock Hardness of {SIS} and {LWE} with small parameters.
\newblock In Ran Canetti and Juan~A. Garay, editors, {\em Advances in
  Cryptology -- {CRYPTO}~2013, Part~I}, volume 8042 of {\em Lecture Notes in
  Computer Science}, pages 21--39, Santa Barbara, CA, USA, August~18--22, 2013.
  Springer, Heidelberg, Germany.

\bibitem[Mut03]{SODA:Muthukrishnan03}
S.~Muthukrishnan.
\newblock Data streams: algorithms and applications.
\newblock In {\em 14th Annual {ACM}-{SIAM} Symposium on Discrete Algorithms},
  pages 413--413, Baltimore, MD, USA, January~12--14, 2003. {ACM-SIAM}.

\bibitem[NY15]{C:NaoYog15}
Moni Naor and Eylon Yogev.
\newblock Bloom filters in adversarial environments.
\newblock In Rosario Gennaro and Matthew J.~B. Robshaw, editors, {\em Advances
  in Cryptology -- {CRYPTO}~2015, Part~II}, volume 9216 of {\em Lecture Notes
  in Computer Science}, pages 565--584, Santa Barbara, CA, USA, August~16--20,
  2015. Springer, Heidelberg, Germany.

\end{thebibliography}

\end{document}